\documentclass{article}
\usepackage[utf8]{inputenc}
\usepackage[a4paper, total={6in, 9.8in}]{geometry}

\usepackage{charter} 

\usepackage[braket, qm]{qcircuit} 
\usepackage{amsmath} 
\usepackage{amsfonts}
\usepackage{bbm} 
\usepackage{graphicx}
\usepackage{xcolor}
\usepackage{bbold}
\usepackage{enumerate}
\usepackage{comment}
\usepackage{authblk}

\usepackage{amsthm}
\newcommand{\ee}{\mathrm{e}}
\newcommand{\ii}{\mathrm{i}}
\newcommand{\rr}{\mathbb{R}}

\newcommand{\gpsr}{\mathrm{gPSR}}

\newcommand{\cpsr}{\mathrm{cPSR}}
\newcommand{\cfd}{\mathrm{cFD}}
\newcommand{\ffd}{\mathrm{fFD}}
\newcommand{\bfd}{\mathrm{bFD}}
\newcommand{\thetab}{{\boldsymbol\theta}}
\newtheorem{theorem}{Theorem}
\newtheorem{corollary}{Corollary}[theorem]

\theoremstyle{definition}
\newtheorem{definition}{Definition}
\newtheorem{remark}{Remark}[definition]

\usepackage[
backend=biber,
style=ieee,
citestyle=numeric-comp
]{biblatex}
\title{Single-component gradient rules \\for variational quantum algorithms}
\author{\small Thomas~Hubregtsen$^1$\thanks{These authors contributed equally to this work.} , 
Frederik~Wilde$^{1 *}$, 
Shozab~Qasim$^1$, 
Jens~Eisert$^{1,2}$}
\affil{\footnotesize ${}^1$ Dahlem Center for Complex Quantum Systems, Freie Universit{\"a}t Berlin, 14195 Berlin, Germany\\
${}^2$ Helmholtz-Zentrum Berlin f{\"u}r Materialien und Energie, 14109 Berlin, Germany}

\date{}

\begin{document}

\maketitle
\vspace*{-1cm}
\begin{abstract}
Many near-term quantum computing algorithms are conceived as variational quantum algorithms, in which parameterized quantum circuits are optimized in a hybrid quantum-classical setup. 
Examples are variational quantum eigensolvers, quantum approximate optimization algorithms as well as various algorithms in the context of quantum-assisted machine learning. 
A common bottleneck of any such algorithm is constituted by the optimization of the variational parameters. 
A popular set of optimization methods work on the estimate of the gradient, obtained by means of circuit evaluations.
We will refer to the way in which one can combine these circuit evaluations as gradient rules.
This work provides a comprehensive picture of the family of gradient rules that vary parameters of quantum gates individually. 
The most prominent known members of this family are the parameter shift rule and the finite differences method. 
To unite this family, we propose a generalized parameter shift rule that expresses all members of the aforementioned family as special cases, and discuss how all of these can be seen as providing access to a linear combination of exact first- and second-order derivatives.
We further prove that a parameter shift rule with one non-shifted evaluation and only one shifted circuit evaluation does not exist, and introduce a novel perspective for approaching new gradient rules.
\end{abstract}

\section{Introduction}
%
Quantum computing has enjoyed a continuous growth in attention, increasingly so since quantum computers are now able to perform tasks promised to be intractable for classical devices~\cite{arute2019quantum, zhong_2020_quantum,Roadmap}.
While these recent experiments impressively showcase the remarkable progress in hardware design and manufacturing, it is unclear on which type of problems near-term quantum computers, without the ability to perform full quantum error correction, can possibly achieve a practical advantage~\cite{preskill2018quantum}.
Still, as more and more small-scale quantum computers become available, the search for a task for which one can reasonably expect a practical quantum advantage is gaining traction. 

Various quantum algorithms are being explored as a potential mechanism for achieving a real-world advantage.
The automotive industry attempts to leverage the quantum approximate optimization algorithm~\cite{farhi_2014_qaoa} for finance and logistics~\cite{qaoa_2021}. 
The aerospace industry investigates quantum machine learning techniques for the simulation of fluid dynamics~\cite{Kyriienko_2020_nde}.
In the field of quantum chemistry, various readings of the variational quantum eigensolver~\cite{peruzzo_2014_variational} are being examined to approximate molecular ground-state energies of molecules~\cite{IBMQuantumChemistry,abbas_2020_power,hempel_2018_chemistry,GogolinVQE,GoogleQuantumChemistry}.
Variational quantum-assisted machine learning algorithms~\cite{bharti_2021_noisy}, or instances of quantum neural networks, are also examined for practical advantage over classical neural networks~\cite{abbas_2020_power}. 
Closely linked are kernel methods~\cite{schuld_2021_quantum, havlivcek2019supervised} and trainable kernel methods~\cite{hubregtsen_2021_trainable_kernels}. 
One can also go a step further and use variational quantum circuits to design other applications of quantum technologies, such as in quantum metrology ~\cite{QCAD,Borregaard}, so that one quantum device is delivering a blueprint for another.

All these algorithms are commonly conceived as
hybrid quantum-classical algorithms~\cite{McClean_2016} in which a quantum circuit depends on a set of free parameters which are typically optimized in an iterative loop controlled by a classical optimizer.
This for good reason, as any reasonable quantum architecture in the foreseeable future has limited depth, with the gates fully dedicated to the task at hand. 
The optimization is done by defining an appropriate cost function as the expectation value of some observable with respect to the output state of the parameterized circuit. 
While this idea appears straightforward, it turned out that achieving suitably efficient and effective classical control of hybrid algorithms is a challenging task, often giving rise a bottleneck for the application of such algorithms. 
This applies in particular to settings in which the evaluation of expectation values is costly, such as in cold atomic architectures, where a single measurement can take some tens of seconds, so that meaningful variational algorithms can only be implemented using the most sophisticated tools of classical control.

While several optimization methods for cost functions of this kind have been studied, it remains an open task to find the optimal method given a specific problem and algorithm~\cite{Benedetti_2019}. In recent years, first-order optimization methods, utilizing the gradient of the cost function, have increasingly gained traction in the field of variational quantum algorithms.
Most notable methods for computing gradient components are finite difference methods for approximation and the parameter shift rule~\cite{li_2017_psr, mitarai_2018_psr, schuld_2018_psr, liu_2018_psr, schuld_2019_psr, crooks_2019_psr, banchi2021measuring, mari_2020_hod} to calculate the components exactly.
Alternative methods entail the quantum natural gradient~\cite{stokes2020quantum}, which calculates the gradient in the distribution space as opposed to the parameter space, and simultaneous perturbation stochastic approximation~\cite{spall1998overview}, which estimates multiple parameters at once.
These approaches can be combined with approaches aimed at
maximizing the information gain from the coherence available to speed up estimation of the expectation values \cite{MinimizingRuntime}.

A first version of the parameter shift rule was first used in the area of control-pulse theory~\cite{li_2017_psr}.
It entailed two circuit evaluations with a shift of $\pi/2$ for the parameter under evaluation. 
Similar approaches have been employed in subsequent improvements~\cite{mitarai_2018_psr, schuld_2018_psr, liu_2018_psr}.
This work was expanded to hold for any unitary gate $U_G=e^{i \theta G}$ whose Hermitian generator $G$ has two distinct eigenvalues in a circuit  $f(\theta) = \bra{\psi}U_G^\dagger(\theta)AU_G(\theta)\ket{\psi}$ with observable $A$~\cite{schuld_2019_psr}.
Gates that have generators with more than two unique eigenvalues can be handled by decomposing the gates~\cite{crooks_2019_psr}. 
A stochastic algorithm for estimating the gradient of any multi-qubit parametric quantum evolution through a mathematically exact formula, without the need for ancillary qubits or the use of Hamiltonian simulation techniques followed~\cite{banchi2021measuring}.
Subsequently, the estimation of higher-order derivatives was proposed, together with a derivation of a formula that allows for gate-independent shifts of opposing angle~\cite{mari_2020_hod}. 
We will use the variable $\gamma$, to indicate the value of a shift that can be freely chosen, e.g., independent from the characteristics of the gate. 
This leads to the current centered parameter shift rule (cPSR) to calculate the analytic gradient $f'(\theta)$ using circuit evaluations as
\begin{align}
    \label{eq:1st_order}
    g_{\cpsr}(\theta, \gamma) :=
    \frac{r[f(\theta + \gamma) - f(\theta - \gamma)]}{\sin(2r\gamma)},
\end{align}
for any $\gamma$ that is not an integer multiple of $\pi/(2r)$, where $2r$ is the difference of the eigenvalues of $G$. 

In this work, we unite all known members of the family of methods that use circuit evaluations to calculate each component of the gradient vector individually, by introducing a generalized parameter shift rule. 
This also yields the variance and bias of all previous methods.  
Furthermore, we close the discussion on the potential existence of a forward or backward parameter shift rule by proving its non-existence. 

We start our work by describing the underlying setting in Section \ref{sec:setting}. 
The generalized parameter shift rule is presented in Section \ref{sec:gpsr}, and show its relation to the finite differences and parameter shift rule. 
Section \ref{sec:nogo} discusses the no-go results for the forward and backward parameter shift rule. We end this work with a discussion of the type of data any of these methods delivers and outlook on future research in Section \ref{sec:discussion}, on how a general picture of gradient rules should be conceived. 
\section{Setting}
\label{sec:setting}

Variational algorithms as considered in this work consist of \emph{parameterized quantum circuits}, mapping $\thetab\mapsto U(\thetab)$, that are defined by the quantum circuit topology and variational parameters $\thetab = (\theta_1,\dots, \theta_n)$.
Note that in the case of embedding data $\boldsymbol{x}$, the circuit can be imagined to have an implicit dependence of the data $(\thetab, \boldsymbol{x})\mapsto U(\thetab, \boldsymbol{x})$, where $\boldsymbol{x}$ can be assumed to be constant for the purpose of this work.
%
In order to optimize the circuit we first need to specify a cost function.
Cost functions are in general not merely expectation values but functions thereof, but for the purposes of the present work it is sufficient to investigate expectation values as we can use the chain rule.
\begin{definition}[Cost function]
    For variational parameters $\thetab= (\theta_1,\dots, \theta_n)$, a parameterized quantum circuit $\thetab\mapsto U(\thetab)$, an observable $A$, and a quantum state vector $\ket{\psi}$ reflecting a pure initial state, a cost function $F:\rr^n\rightarrow \rr$ is defined as
    \begin{equation}
        F(\thetab) := \bra{\psi} U(\thetab) A U^\dagger(\thetab)
        \ket{\psi}.
    \end{equation}
    A single-shot estimator $\hat{F}(\thetab)$ for the cost function $F(\thetab)$ corresponds to 
    a single measurement of the Hermitian observable $A$ when the system
    is described by the state vector $U(\thetab)\ket{\psi}$.
    Following these lines, the $n$-shot estimator $\hat{F}^{(n)}(\thetab)$ is the sample mean of $n$ i.i.d. samples of the single-shot estimator.
\end{definition} 
We then complete the variational algorithm by searching for updates to the parameters $\thetab$ such that the cost function is minimized.
In the entirety of this work, we will focus on optimizing the cost function $F$ by estimating its gradient $\nabla_\thetab F$.
Gradient rules are strategies that help estimate the gradient, leveraging carefully chosen circuit evaluations exploiting symmetries in the problem.
These methods dictate shifts for the parameters of the circuit, and how to combine the resulting cost function measurements to form the gradient. 

\begin{definition}[Gradient rule]
    To estimate the $i$-th component of the gradient at the parameter point $\thetab \in \rr^n$, a gradient rule $g$ requires the evaluation of the cost function $F$ at shifted parameter points $\thetab + \boldsymbol{s}_{1}, \ldots, \thetab + \boldsymbol{s}_{k}$ (depending on $i$) 
    and describes how the resulting measurements form this gradient component, i.e.
    \begin{equation}
        g(\thetab, \boldsymbol{s}_1, \ldots, \boldsymbol{s}_k)
        = \frac{\partial F}{\partial\theta_i}.
    \end{equation} 
    Corresponding to the gradient rule $g$, we define the estimator $\hat{g}$ where all cost function evaluations $F(\thetab + \boldsymbol{s})$ are replaced by single-shot measurement estimates $\hat{F}(\thetab + \boldsymbol{s})$.
\end{definition}
By applying gradient rules multiple times, one can recover higher order derivatives such as the Hessian.
In the next section, we will exclusively focus on the family of gradient rules that alter each parameter individually. 
We will refer to these as single-component gradient rules.
As an example, in the case of the original parameter shift rule~\cite{schuld_2018_psr} for the $i$-th component, we have $\boldsymbol{s}_{1} = \pi/(4r)\boldsymbol{e}_i$ and $\boldsymbol{s}_{2} = -\pi/(4r)\boldsymbol{e}_i$, where $\boldsymbol{e}_i$ is the $i$-th canonical basis vector.
\section{Single-component gradient rules}
\label{sec:gpsr}

For the scenario of gradient rules that alter the parameters individually, we can now define a single-component cost function for the $i$-th component of the parameter vector, defined as
\begin{align}
    f(\theta) :=  F(\theta_1,\dots,\theta_{i-1}, \theta, \theta_{i+1},\dots, \theta_n) = \bra{\psi} U(\thetab) A U(\thetab)^\dagger
        \ket{\psi},
\end{align}
where all parameters, except $\theta_i$ are fixed.
This gives rise to a family of single-component gradient rules giving access to $f' = \partial f / \partial\theta$.
To rigorously state the (generalized) parameter shift rule, we define $r$-gates and the set of circuits parameterized by them.

\begin{definition}[$r$-gate]
    A parameterized gate $\ee^{-\ii \theta G}$ is called an $r$-gate if its Hermitian generator $G$ has exactly two eigenvalues $e_0$ and $e_1$, such that $e_0 - e_1 = 2r$.
\end{definition}

\begin{definition}[Set of $r$-gate parameterized expectation values]
    \label{def:set-of-expec-vals}
    Let $q \geq 1$ be the number of qubits in the system at hand.
    Then, the set $\mathcal{F}_r$ of all $r$-gate parameterized expectation values on this system is defined as follows.
    A single-parameter family of expectation values $f:[0, 2\pi/r) \rightarrow \mathbb{R}$ is contained in $\mathcal{F}_r$ if and only if there exists a 
    state vector $\ket{\psi} \in \mathcal{H} := \mathbb{C}^{\otimes q}$, an $r$-gate $\ee^{-\ii\theta G}$, and a Hermitian operator $A$ acting on $\mathcal{H}$, such that for all $\theta\in 
    [0,2\pi/r)$
    \begin{equation}
        f(\theta) = \bra{\psi} \ee^{\ii \theta G} A \ee^{-\ii \theta G}\ket{\psi}.
    \end{equation}
\end{definition}
\begin{remark}[Convention for the spectrum of generators]
    We assume $e_0 = -e_1 = r$ for the purposes of this work~\cite{schuld_2019_psr}.
    If this was not the case, we could simply add a phase shift $\ee^{\ii\theta (e_0 - e_1)/2}$ to the gate $\ee^{-\ii\theta G}$, which would have no effect on the expectation value.
    With this assumption, we can express the gate as
    \begin{equation} \label{eq:euler}
        \ee^{-\ii \theta G} = \cos(r\theta) - \ii \frac{G}{r} \sin(r\theta),
    \end{equation}
    where we have exploited the fact that $G^2 = r^2\mathbb{1}$.
\end{remark}

In particular, we will be now be taking a closer look at and generalizing the \emph{centered, backward} and \emph{forward finite difference rules} (cFD, bFD and fFD respectively), as well as the \emph{centered parameter shift rule}~(cPSR)~\cite{li_2017_psr, mitarai_2018_psr, schuld_2018_psr, liu_2018_psr, schuld_2019_psr, crooks_2019_psr, banchi2021measuring, mari_2020_hod}. 
In the above language, they can be stated as
\begin{align}
    g_{\cfd}(\theta, h) &:= \frac{f(\theta + h) - f(\theta - h)}{2h} \simeq f'(\theta), \label{eq:cfd} \\
    g_{\bfd}(\theta, h) &:= \frac{f(\theta) - f(\theta - h)}{h} \simeq f'(\theta), \label{eq:bfd} \\
    g_{\ffd}(\theta, h) &:= \frac{f(\theta + h) - f(\theta)}{h} \simeq f'(\theta) ,\label{eq:ffd} \\
    g_{\cpsr}(\theta, \gamma) &:= 
    \frac{r[f(\theta + \gamma) - f(\theta - \gamma)]}{\sin(2r\gamma)} = f'(\theta). 
\end{align}
Here, $h$ and $\gamma$ represent the shift of the parameter, where $h$ is small compared to the period of $f$ while $\gamma$ can take any value except $k\pi/(2r)$ for any $k\in\mathbb{Z}$. 
It is important to note that while the first three rules aim at approximating the first derivative, the last rule provides an exact expression on the level of expectation values
for all $\gamma\in \mathbb{R}$.
In the following definition, we will present the generalized parameter shift rule, which we show in the subsequent proof to express the first- and second-order derivatives.
In the subsequent corollaries, we show how the generalized parameter shift rule relates to the gradient rules in this family, as well as their relation to the (higher-order) derivatives.  
This will also include exact expressions for the bias, whereas previously the asymptotic relation in equations (\ref{eq:cfd}-\ref{eq:ffd}) could only be approximated.

\begin{definition}[Generalized parameter shift rule]
    The generalized parameter shift rule (gPSR) is defined as 
    \begin{align}
        g_\gpsr(\theta, \gamma_1, \gamma_2) := r \left[ f(\theta + \gamma_1) - f(\theta - \gamma_2) \right],
    \end{align}
    for $\gamma_1,\gamma_2\in \mathbb{R}$, where $f \in \mathcal{F}_r$ is an $r$-gate parameterized expectation value (see Definition \ref{def:set-of-expec-vals}).
\end{definition}

\begin{theorem}[Deriving rules from the generalized parameter shift rule]
    Given two shifts $\gamma_1, \gamma_2 \in \rr$, the gPSR as defined above relates to the first and second derivative of $f$ as 
    \begin{eqnarray} 
        \label{eq:gpsr_equality_1}
        g_\gpsr(\theta, \gamma_1, \gamma_2)
        = \frac{ \sin(2r \gamma_1) + \sin(2r \gamma_2)}{2} f'(\theta) - \frac{\cos(2r \gamma_1) - \cos(2r \gamma_2)}{4r} f''(\theta). 
    \end{eqnarray}
\end{theorem}

\begin{proof}
    We first need to decompose the $r$-gates in $\gamma$ and $\theta$-dependent gates as
    \begin{equation}
    \ee^{-\ii (\theta+\gamma) G} = \ee^{-\ii \theta G} \ee^{-\ii \gamma G}.
     \end{equation} 
    The gate depending on $\theta$ can be absorbed to simplify the equation, by defining $\ket{\phi(\theta)} := \ee^{-\ii \theta G}\ket{\psi}$. 
    In contrast, the gate depending on $\gamma$ needs to be decomposed using Euler's identity from equation (\ref{eq:euler}), as done previously in Ref.~\cite{crooks_2019_psr}, but here with unequal shifts $\gamma_1$ and $\gamma_2$, in order
    to get
    \begin{align}
        g_\gpsr(\theta, \gamma_1, \gamma_2) 
        = r \Big\lbrace \bra{\phi(\theta)} \left[ \cos(r\gamma_1)\mathbb{1}+i\frac{a}{r}\sin(r\gamma_1)G \right] &A 
        \left[ \cos(r\gamma_1)\mathbb{1}-i\frac{a}{r}\sin(r\gamma_1)G \right] \ket{\phi(\theta)} \nonumber \\
        - \bra{\phi(\theta)} \left[ \cos(r\gamma_2)\mathbb{1}+i\frac{a}{r}\sin(r\gamma_2)G \right] &A 
        \left[ \cos(r\gamma_2)\mathbb{1}+i\frac{a}{r}\sin(r\gamma_2)G \right] \ket{\phi(\theta)} \Big\rbrace. 
    \end{align}
    We can now expand and reorganize this equation into six terms. 
    Note that $\bra{\phi(\theta)}A\ket{\phi(\theta)}=f(\theta)$ and 
    $\bra{\phi(\theta)}GAG\ket{\phi(\theta)}/r^2=f(\theta + \pi/(2r))$ together can be combined into the second derivative using the double-angle formula $\cos^2(x) = (\cos(2x)+1)/2$.
    The remaining terms involving $\bra{\phi(\theta)} [\ii G,A]   \ket{\phi(\theta)}$ form the first derivative. This gives
    \begin{align}\label{eq:gPSR_derivatives}
        g_\gpsr(\theta, \gamma_1, \gamma_2) 
        &= r \cos^2(r \gamma_1) \bra{\phi(\theta)} A \ket{\phi(\theta)} 
        -r \cos^2(r \gamma_2) \bra{\phi(\theta)} A \ket{\phi(\theta)}  \nonumber \\
        &~~~+ \frac{1}{r} \sin^2(r \gamma_1)\bra{\phi}  G A G \ket{\phi(\theta)} 
        - \frac{1}{r} \sin^2(r \gamma_2)\bra{\phi(\theta)}  G A G \ket{\phi(\theta)}
        \nonumber \\
        &~~~+ \cos(r \gamma_1)\sin(r \gamma_1) \bra{\phi(\theta)} [\ii G,A]   \ket{\phi(\theta)}
        + \cos(r \gamma_2)\sin(r \gamma_2)  \bra{\phi(\theta)} [\ii G,A]   \ket{\phi(\theta)} \nonumber \\
        &= \frac{ \sin(2r \gamma_1) + \sin(2r \gamma_2)}{2} f'(\theta) - \frac{\cos(2r \gamma_1) - \cos(2r \gamma_2)}{4r} f''(\theta). 
    \end{align}    
\end{proof}

\begin{corollary}[Uses of the generalized PSR]
    The generalized PSR, when multiplied with a pre-factor and for the listed values for shift angles $\gamma_1$ and $\gamma_2$, can express any higher-order derivative. 
\end{corollary}

\begin{proof}
This statement immediately follows from equation (\ref{eq:gPSR_derivatives}) using shift angles $\gamma_1, \gamma_2 = \gamma$, resulting in
\begin{align}
    f'(\theta) &=  \frac{g_\gpsr(\theta, \gamma, \gamma)}{\sin(2r\gamma)} \label{psr}
 \end{align}
 for the first derivative and
\begin{align}  
    f''(\theta) &= 2\, g_\gpsr\left(\theta, \frac{\pi}{2r}, 0\right)
\end{align}
for the second derivative.
Note that equation (\ref{psr}) is the standard cPSR. 
For arbitrary $n\in \mathbb{N}$,
higher order derivatives of $f$ satisfy
 \begin{align}
        & f^{(n+2)} = -\frac{1}{4r^2} f^{(n)},
\end{align}
where $f^{(n)}$ refers to the $n$-th derivative of $f$. This is an important insight, in that it
shows that all expressions involving 
derivates of arbitrary order can be written as a
linear combination of first and second derivatives
of the cost function $f$, which in turn is the
data that any of the above rule exactly
delivers on the level of expectation values.
All higher-order derivatives have been previously identified as following from applying the centered PSR multiple times~\cite{mari_2020_hod}.
\end{proof}

\begin{corollary}[Further uses of the generalized PSR]
    The generalized PSR, multiplied by a known pre-factor and with shift angles $\gamma_1=\gamma_2$ is equivalent to the centered finite differences method, and provides an exact expression for the bias and variance of the centered finite differences.
\end{corollary}

\begin{proof}
    By setting $\gamma_1 = \gamma_2 = h$ we have \begin{align}
    g_\gpsr(\theta, h, h) 
    &= 2rh \, g_\cfd(\theta, h),
    \end{align}
    for the centered finite differences, while the bias of the centered finite differences becomes
    \begin{align} 
    \text{Bias}(\hat{g}_\cfd(\theta, h)) 
    &=  
    \mathbbm{E}\left[ \hat{g}_\cfd(\theta, h) \right] - f'(\theta)
    = 
    \frac{g_\gpsr(\theta, h, h)}{2rh} - f'(\theta)\\
    &=       \left[\frac{\sin(2rh)}{2rh} -1 \right] f'(\theta) .
    \end{align}
    The variance was already known to be
    \begin{align}\label{eq:cFD-variance}
        \text{Var}(\hat{g}_\cfd(\theta, h)) 
        &= \frac{\sigma_1^2(\theta + h) + \sigma_1^2(\theta - h)}{4h^2},
    \end{align}
    where $\sigma_1^2$ is the one-shot variance~\cite{mari_2020_hod}.
\end{proof}

\begin{corollary}[Backward and forward finite differences from the generalized PSR]
    The generalized PSR, multiplied by a known pre-factor and with shift angles $0$ and $h$, is equivalent to both the backward and forward finite differences, and provides an exact expression for the bias and variance of these finite difference methods.
\end{corollary}

\begin{proof}
    For the backward FD we set $\gamma_1 = 0$ and $\gamma_2 = h$, to get
    \begin{align}
        g_\gpsr(\theta, 0, h) 
        &= h\, g_\bfd(\theta, h), \\
        \text{Bias}(\hat{g}_\bfd(\theta,h)) 
        &=
        \mathbbm{E}\left[ \hat{g}_\bfd(\theta,h) \right] - f'(\theta)
        = 
        \frac{g_\gpsr(\theta, 0, h)}{h} - f'(\theta) \nonumber \\
        &= 
        \left[ \frac{\sin(2rh)}{2h} - 1 \right] f'(\theta) + \frac{\cos(2rh)-1}{4rh} f''(\theta), \\
        \text{Var}(\hat{g}_\bfd(\theta,h)) 
        &= \frac{\sigma_1^2(\theta) + \sigma_1^2(\theta - h)}{4h^2}.
    \end{align}
    By choosing $\gamma_1 = h$ and $\gamma_2 = 0$ the same relation for bias and variance is obtained for the forward FD method.
\end{proof}

The here presented gPSR unites all known forms of the finite differences and the centered PSR to the analytic gradient, thereby uniting all known methods in the family of single-component gradient rules and settling the discussion on trade-off between these approximate and exact methods. 
For the remaining two unknown forms, the backward and forward PSR, we provide a no-go theorem in the the following section.
\section{No-go result for a forward and backward parameter shift rule}
\label{sec:nogo}

It is reasonable to expect near-term quantum devices to allow for wider and deeper circuits in the pursuit of achieving a practical quantum advantage, thereby resulting in the need to optimize a large number of parameters in the variational algorithm.
Evaluating the gradient for $n$ parameters requires $2n$ cost function evaluations per gradient vector evaluation with the parameter shift rule.
As a hypothetical backward or forward parameter shift rule, in the spirit of the backward and forward finite differences, would require only $n+1$ evaluations of the cost function, it is desirable to find such a rule.
Here, we show that it is impossible for such a rule to exist when one considers the set of all cost functions expressed through variational quantum circuits.

\begin{theorem}[Impossibility of a backward and forward parameter shift rules]
    For all $0 < \gamma < \pi/r$ and $r \in \mathbb{R}$ there exists no function $g:\mathbb{R} \times \mathbb{R} \rightarrow \mathbb{R}$ with the property that 
    \begin{equation}
        g[f(\theta), f(\theta + \gamma)] = f' (\theta)
    \end{equation}
    for all circuit expectation values $f \in \mathcal{F}_r$ (see Definition \ref{def:set-of-expec-vals}) and all $\theta \in [0, \pi/r)$.    
\end{theorem}

\begin{proof}
    Let $\zeta \in \mathbb{R}$ and $0 < \gamma < \pi/r$ be arbitrary but fixed.
    We pick an $f\in \mathcal{F}_r$ along with some $r$-gate $F$ where
    \begin{align}
        f(\theta) &= \bra{\psi} \ee^{\ii\theta G} A \ee^{-\ii\theta G} \ket{\psi},
    \end{align}
    and the  conditions
    \begin{align}
        c_\perp \ket{\zeta_\perp} := \ee^{-\ii\gamma F}\ket{\zeta} - \langle\zeta\vert\ee^{-\ii\gamma F}\vert\zeta\rangle\ket{\zeta} &\neq 0, \\
        \label{eq:GF-commutator}
        \bra{\zeta} [G-F, A] \ket{\zeta} &\neq 0,
    \end{align}
    are fulfilled using the short-hand $\ket{\zeta} := \ee^{-\ii \zeta G}\ket{\psi}$,
    where we have imposed $\ket{\zeta_\perp}$ to have unit norm.
    For most choices of $G$, $F$, $A$, and $\ket{\psi}$ these conditions will be fulfilled. However, it suffices to show that one such choice exists: For instance, choose $G=X$, $F=(Y+Z)/\sqrt{2}$, $A=Y$, and $\ket{\psi} = \ee^{\ii\zeta X}\ket{0}$.
    Now we can define a second element of $\mathcal{F}_r$ as
    \begin{align}
        \tilde{f}(\theta) &:= \bra{\zeta} \ee^{\ii(\theta - \zeta) F} B \ee^{-\ii(\theta - \zeta) F} \ket{\zeta},
    \end{align}
    with the observable
    \begin{align}
        B := A - \frac{\bra{\zeta}\ee^{\ii\gamma G} A\ee^{-\ii\gamma G}\ket{\zeta} - \bra{\zeta}\ee^{\ii\gamma F} A\ee^{-\ii\gamma F}\ket{\zeta}}{\vert c_\perp\vert^2}
        \ket{\zeta_\perp}\!\bra{\zeta_\perp}.
    \end{align}
    One can now verify that with this choice, we have
    \begin{align}
        \tilde{f}(\zeta) &= \bra{\zeta}A\ket{\zeta} = f(\zeta) ,\\
        \tilde{f}(\zeta + \gamma) &= \bra{\zeta}\ee^{\ii\gamma G} A\ee^{-\ii\gamma G}\ket{\zeta} = f(\zeta + \gamma).
    \end{align}
    Meanwhile, the derivatives at
    the point $\zeta$ are given by
    \begin{align}
        f'(\zeta) &= \bra{\zeta}[\ii G, A]\ket{\zeta} ,\\
        \tilde{f}'(\zeta) &= \bra{\zeta} [\ii F, A] \ket{\zeta},
    \end{align}
    which are unequal by the assumption in Eq.~(\ref{eq:GF-commutator}).
    Therefore, the function $g$ discussed here, as described in the statement of the theorem, would receive the same inputs for $f$ and $\tilde{f}$ at $\theta=\zeta$ and $\theta=\zeta+\gamma$, while having to return different outputs.
    Hence, such a function $g$ cannot exist.
\end{proof}

\section{Discussion}
\label{sec:discussion}

Variational quantum algorithms hold great potential for practical quantum applications. 
These algorithms are limited by their scaling in terms of number of variational parameters $n$, due to the increase in optimization difficulty.
Gradient methods have gained significant attention over the past years as key method for optimization, and come in various forms. 
In our work, we focused on the single-component gradient methods, including multiple variations of the finite differences (FD) and the parameter shift rule (PSR).
The key lesson to be learned from our work, is that any of the above methods deliver linear combinations of higher-order derivatives, which can be reduced back to a linear combination of the first and second derivatives, of the cost function 
\begin{equation}
g(\theta, \boldsymbol\gamma)
= a({\boldsymbol\gamma})
f'(\theta) +
b({\boldsymbol\gamma})
f''(\theta)
\end{equation}
with known functional dependence of the coefficients.
The centered parameter shift rule (cPSR) solely targets the first derivative directly - the other methods however deliver instances of such linear combinations. 
We have formalized this, and proposed the generalized parameter shift rule (gPSR).
This gPSR relates, in an exact manner, to all versions of the FD method, as well as the cPSR. 
Furthermore, it allowed us to calculate the bias and variance for all previously mentioned methods. 
To fully cover the family of single-component gradient rules, we have also presented a proof that neither a backward, nor forward parameter shift rule can exist. 

The notion that each method delivers a linear combination of first- and second-order derivatives, points at an opportunity for future work to explore how these linear combination can be exploited, when addressing the task of classical control of variational quantum algorithms. 
Traditional second order optimization methods are known to
outperform first order methods, in that they allow for the assignment of better directions of updates, as well as better step sizes; allowing for jumps directly to the minimum of the local quadratic approximation \cite{Second}. 
Here, the situation is slightly different and somewhat unusual from the perspective of second order optimization - in that first and second derivatives in a fixed combination are natively provided by the data, and most resource efficient  approaches should use those native data at hand when addressing a control problem.

This line of reasoning can also be extended to the multi-component gradient rules, where multiple parameters are updated simultaneously.
A natural first step would be to explore the existence of an exact multi-parameter shift rule, aimed at improving the number of circuit evaluations per chosen parameter.
General rules of this type would give rise to scalar products with gradients and Hessians of the cost function. These can be exploited in suitable second order methods.
Approaches of multi-parameter optimization would also open up the possibility of exploiting further structure in the gradient, inspired from the domain of signal processing - this includes resorting to sparse and low-rank structures \cite{CompressedSensingIntroCandes}.
However, in line with results from this work, we believe one should look beyond this, and aim to exploit the fact that multiple cross-derivatives between parameters can be accessed as a linear combination; just as we showed for the single-component setting. 
For instance, as higher-order derivatives can be expressed in terms of first and second derivatives, the cost function can also be expanded in \emph{all} parameters simultaneously.
Here the $l$-th derivative, represented by a rank-$l$ tensor, can be computed at a cost significantly lower than $n^l$, where $n$ is the number of parameters.
Moreover, the series has finite length.

In particular, for a setting with two parameterized gates, the circuit evaluation of the expectation value of $f(\theta_1 + \gamma_1, \theta_2 + \gamma_2)$ can not only express first- and second-order derivatives for each parameter of $(\theta_1, \theta_2)\mapsto f(\theta_1, \theta_2)$, but also mixed (higher-order) derivatives.
Given this access, one could envision a gradient rule that yields a linear combination of a subset of these derivatives, and an optimization method that works directly on the outcome of this gradient rule.
One example of such an approach would be the following set of circuit evaluations, which yield a linear combination of the Hessian elements as
\begin{align}
    & f(\theta_1 + \gamma_1, \theta_2 + \gamma_2) + f(\theta_1 - \gamma_1, \theta_2 - \gamma_2) - 
    2\sin^2(r\gamma_1)\sin^2(r\gamma_2)\, f\left(\theta_1 + \frac{\pi}{2}, \theta_2 + \frac{\pi}{2}\right) \nonumber \\
    &\quad = 2 \sin(2r\gamma_1)\sin(2r\gamma_2) \frac{\partial}{\partial \theta_1}\frac{\partial}{\partial \theta_2}f(\theta_1, \theta_2)
    + \frac{1}{4}\cos^2(r\gamma_2) \left(5 - 3\cos(2r\gamma_1)\right)
    \frac{\partial^2}{\partial \theta_1^2}f(\theta_1, \theta_2) \nonumber \\
    & \quad\quad+ \frac{1}{4} \cos^2(r\gamma_1) \left(5 - 3\cos(2r\gamma_2)\right)
    \frac{\partial^2}{\partial \theta_2^2}f(\theta_1, \theta_2).
\end{align}
When using this as input for the optimization, no classical access to the individual components of the Hessian are required; thereby reducing the required number of evaluations of the quantum circuit.
Such a setup could give rise to pseudo-Newton methods that make use of this linear combination, directly in its optimization process. 
A comprehensive solution would extend this methodology, and economically estimate the full gradient over all parameters, based on suitable control patterns over the variational parameters. 
It is our hope that the present work stimulates such research endeavours in near-term quantum computing.

\section*{Acknowledgements}
We would like to thank Johannes~Jakob~Meyer, Paul~K.~Faehrmann, and Ryan~Sweke for helpful discussions. This work has been supported by the BMWi (PlanQK) and the BMBF (Hybrid and FermiQP). It has also received funding from the DFG under Germany's Excellence Strategy – 
MATH+: The Berlin Mathematics Research Center, 
EXC-2046/1 – project ID: 390685689,
and the CRC 183, as well as the European Union's 
Horizon 2020 research and innovation programme under 
grant agreement No.\ 817482 (PASQuanS).

\printbibliography

\end{document}